\documentclass{article}
\usepackage{graphicx}
\usepackage{cite}
\usepackage{amsmath}
\usepackage{amsthm}

\newtheorem{theorem}{Theorem}
\newtheoremstyle{mystyle}{}{}{\itshape}{}{\bfseries}{.}{ }%
  {\thmname{#1}\thmnumber{\@ifnotempty{#1}{ }\@upn{#2}}%
    \thmnote{ {\bfseries(#3)}}}
\theoremstyle{mystyle}
\newtheorem*{definition}{Definition}

\begin{document}

\title{Counting Independent Terms in Big-Oh Notation}

\author{Fabiano de S. Oliveira$^1$\thanks{Corresponding author (fabianoo@gmail.com).}\\
Valmir C. Barbosa$^2$\\
\\
$^1$Instituto de Matem\'atica e Estat\'\i stica\\
Universidade do Estado do Rio de Janeiro\\
Rua S\~ao Francisco Xavier, 524, sala 6019B\\
20550-900 Rio de Janeiro - RJ, Brazil\\
\\
$^2$Programa de Engenharia de Sistemas e Computa\c c\~ao, COPPE\\
Universidade Federal do Rio de Janeiro\\
Caixa Postal 68511\\
21941-972 Rio de Janeiro - RJ, Brazil}

\date{}

\maketitle

\begin{abstract}
The field of computational complexity is concerned both with the intrinsic hardness of computational problems and with the efficiency of algorithms to solve them. Given such a problem, normally one designs an algorithm to solve it and sets about establishing bounds on its performance as functions of the algorithm's variables, particularly upper bounds expressed via the big-oh notation. But if we were given some inscrutable code and were asked to figure out its big-oh profile from performance data on a given set of inputs, how hard would we have to grapple with the various possibilities before zooming in on a reasonably small set of candidates? Here we show that, even if we restricted our search to upper bounds given by polynomials, the number of possibilities could be arbitrarily large for two or more variables. This is unexpected, given the available body of examples on algorithmic efficiency, and serves to illustrate the many facets of the big-oh notation, as well as its counter-intuitive twists.

\end{abstract}

\newpage
\section{Introduction}

Computer algorithms require resources in order to run, most notably, time (the number of steps they must go through before  termination) and space (the number of memory cells where their input and intermediate results are to be stored). In general, any given run of an algorithm may require a different amount of such resources even if the algorithm is fully deterministic, since both time and space usage depend heavily on the input to the algorithm.

This dependence on the input is quantified by means of what here we call an algorithm's \emph{variables}, that is, numbers that explicitly or implicitly are part of any input to the algorithm and can affect its resource requirements for computing on that particular input. For instance, while the number of steps required by some algorithms for sorting an array of integers depends on the size of the array (this being a piece of information that probably is an explicit part of the input but in any case could easily be derived from it), the greatest integer in the array does not affect the number of steps. The latter holds under the assumption that each individual integer can be stored in a single processor register, which seems reasonable given that it is true of any modern processor for integers up to $4$ billion.

The amount of time or space required by an algorithm is expressed as a function of its variables. In the sorting example, both time and space are nondecreasing functions of the single variable representing the number of integers in the input. As it happens, though, in most cases determining an algorithm's exact time or space function is a rather difficult task. To circumvent some of this difficulty while still allowing for meaningful statements about the algorithm's performance to be made, the commonly accepted practice has been to express such functions, through the so-called \emph{big-oh} notation, in asymptotic terms.

\begin{definition}[big oh, one variable]
Let $f(n)$ and $g(n)$ be real functions. We say that $f(n) = O(g(n))$ if there exist positive constants $c$ and $n_0$ such that $f(n) \leq c g(n)$ for all $n \geq n_0$. 
\end{definition} 

The big-oh notation is in fact more than simply a notation, since the elements that go in its definition allow for a cleaner view of an algorithm's inner workings as far as its resource requirements are concerned. Thus, instead of seeking to determine, say, the time function $f(n)$ of an algorithm, what one does is to look for some $g(n)$ that, for sufficiently large $n$, is proportional to an upper bound on $f(n)$. If such an upper bound is ``tight'' (i.e., if it does reflect the time requirement of the worst-case inputs to the algorithm), then several conclusions can be drawn with the help of $g(n)$. For example, if algorithms $A$ and $B$ admit tight upper bounds on their time functions, respectively $O(n^2)$ and $O(n \log n)$, then the worst-case performance of algorithm $B$ is preferable to that of $A$ for sufficiently large $n$.

It often happens for an algorithm's time and space functions to depend on more than one variable, each representing a different aspect of the inputs to the algorithm. This occurs routinely in the case of algorithms on graphs, since very commonly such algorithms' resource requirements are influenced by both the graph's number of vertices and its number of edges, but it may also occur more indirectly. As an example, consider an algorithm $A$ for sorting a set of $n$ integers in $O(n^2)$ time. Suppose further that an input to $A$ comes in the form of two disjoint sets, say $X$ and $Y$, respectively of sizes $x$ and $y$, and that for reasons that have to do with differentiating elements in $X$ from those in $Y$ in some application, the natural way to express the running time of $A$ is by making explicit use of $x$ and $y$, as in $f(x+y)$, rather than coalescing them as a sum into the variable $n$ and using $f(n)$ instead. Proceeding in this way would lead to $f(x+y) = O((x+y)^2) = O(x^2+2xy+y^2)$, but clearly these expressions are clamoring  for the big-oh notation to be extended to the two-variable case.

\begin{definition}[big oh, two variables]
Let $f(n_1,n_2)$ and $g(n_1,n_2)$ be real functions. We say that $f(n_1,n_2) = O(g(n_1,n_2))$ if there exist positive constants $c$, $n_{1,0}$, and $n_{2,0}$ such that $f(n_1,n_2) \leq c g(n_1,n_2)$ for all $n_1,n_2\geq 0$ satisfying $n_1 \geq n_{1,0}$ or $n_2 \geq n_{2,0}$. 
\end{definition}

With the extended definition in hand, we can express the algorithm's time function as $f(x,y)=O(x^2+2xy+y^2)$ and finally see that, in reality, $f(x,y)=O(x^2+y^2)$. This is so because $2xy=O(x^2)$ for all valuations of $x$ and $y$ in which $y\leq x$ and $2xy=O(y^2)$ for those in which $x\leq y$.

This example is interesting also in that it highlights the advantage of being concise when using the big-oh notation: saying that $f(x,y) = O(x^2+y^2)$ is no less informative than saying that $f(x,y) = O(x^2+2xy+y^2)$, but is more useful to a potential user of the algorithm in question. Motivated by this observation, we equate conciseness with irreducibility, the latter defined as follows. First, let a \emph{term} be the product of single-variable functions; e.g., $x^2$, $2xy$, and $y^2$ are all terms. Given the $k$ terms $T_1,\ldots,T_k$, we say that term $T_i$ is \emph{independent} with respect to the sum $S=T_1+\cdots+T_k$ if it is not asymptotically bounded from above in proportion to $S-T_i$, that is, if $T_i = O(S-T_i)$ does not hold. We say that $S$ is \emph{irreducible} if all of $T_1,\ldots,T_k$ are independent. In the example, $x^2+y^2$ is irreducible but $x^2+2xy+y^2$ is not.

In this article we concern ourselves with functions that, like $S$, can be written as a sum of terms. Our goal is to answer a specific question motivated by the following practical application. Suppose we are given the executable code for some program, along with the list of variables affecting its performance, but no further information (no source code, no time or space function, not even their big-oh forms). Suppose further that, before putting such code to use, we are tasked with estimating a bound on its time (or space) function as concisely as possible, via the big-oh notation, over a given range of the variables. If we were allowed to restrict our search to those time functions that can be expressed as a sum of terms, then knowing beforehand how many terms there can be in a concise representation thereof would be of great help. The question that interests us is then, how many terms can an irreducible sum-of-terms function have?

\section{One variable, and beyond}

Answering this question becomes easier if we make further assumptions, now regarding the nature of the single-variable factors that make up a term. If the sum-of-terms function in question is itself a function of a single variable, then the further assumption is quite reasonable and refers to restricting those factors to be one of the functions that commonly appear in algorithmic analysis: polynomial, polylogarithmic, logarithmic, and exponential. This given, the question is answered quite simply in the single-variable case, in which no sum of at least two terms constitutes an irreducible function.

We see this more clearly by following Hardy~\cite{Hardy}, who defined the class $\mathcal{L}$ of \emph{logarithmico-exponential functions} to be the one comprising the following functions: 
\begin{itemize}
\item $f(n) = a$, for any real constant $a$;
\item $f(n) = n$;
\item $f(n)-g(n)$, if $f(n), g(n) \in \mathcal{L}$;
\item $e^{f(n)}$, if $f(n) \in \mathcal{L}$;
\item $\ln f(n)$, if $f(n) \in \mathcal{L}$ and, for some constant $n_0$, $f(n) > 0$ for all $n \geq n_0$.
\end{itemize}
As it turns out, any function arising naturally when analyzing an algorithm belongs to $\mathcal{L}$~\cite{ConcreteMath}. For example, $4n^3 + (\log n)^5 + 2^{n^2}$ can be seen to be in $\mathcal{L}$ by using the defining conditions for $\mathcal{L}$ as follows:
\begin{itemize}
\item $f(n)+g(n) = f(n)-(0-g(n)) \in \mathcal{L}$ if $f(n),g(n) \in \mathcal{L}$;
\item $f(n)g(n) = e^{\ln f(n)+\ln g(n)} \in \mathcal{L}$ if $f(n),g(n) \in \mathcal{L}$; 
\item $f(n)^k = \prod_{i=1}^k f(n) \in \mathcal{L}$ if $f(n) \in \mathcal{L}$;
\item $k f(n) = \sum_{i=1}^k f(n) \in \mathcal{L}$ if $f(n) \in \mathcal{L}$;
\item $2^{f(n)} = e^{\ln 2 f(n)} \in \mathcal{L}$ if $f(n) \in \mathcal{L}$.
\end{itemize}

An important consequence of Hardy's work is that, given any two functions $f(n)$ and $g(n)$ in $\mathcal{L}$, we have $f(n) = O(g(n))$ or $g(n) = O(f(n))$. Therefore, if $f(n)$ is the sum of at least two terms, each one in $\mathcal{L}$, then $f(n)$ is not irreducible.

And how about sum-of-terms functions having more than one variable? As noted above, multiple variables are a common occurrence in graph algorithms, whose time and space functions often depend on both $n$ and $m$, the graph's numbers of vertices and edges, respectively. In fact, some of the best algorithms for numerous graph problems have time functions bounded by sums of two or three terms, often depending on more variables than simply $n$ and $m$, as shown in Table~\ref{GraphAlgs}. What we see in the table are counts of how many algorithms, as reported in a portion of \cite{Schrijver03}, have time-function bounds with a certain number of terms on a certain number of variables. For example, $65$ of the reported bounds on two variables have one single term and $14$ have two, but none has three or more terms.

\begin{table}[t]
\caption{Number of time-function bounds reported in the ``Complexity survey'' subsections of ~\cite{Schrijver03}, considering the books' Parts I and II only.}
\label{GraphAlgs}
\centering
\begin{tabular}{c c c c}
\hline
Number &\multicolumn{3}{c}{Number of terms} \\
\cline{2-4}
of variables & 1 & 2 & 3 \\
\hline
1 & 52 & 0 & 0 \\
2 & 65 & 14 & 0 \\
3 & 48 & 33 & 2 \\
4 & 5 & 14 & 0 \\
5 & 3 & 2 & 0 \\
\hline
Total & 173 & 63 & 2 \\ 
\hline
\end{tabular}
\end{table}

One might then wonder if two is the maximum number of terms whose sum is irreducible in the case of two variables. That, however, is not the case. To see this, consider the function $f(x,y) = x^2+y^2+(xy)^{3/2}$. The first term in this function is independent, since it does not hold that $x^2 = O(y^2+(xy)^{3/2})$, as seen by simply fixing $y=c$ for any positive constant $c$. The case of the second term is entirely analogous. As for the third term, set $x=y$ to conclude that $(xy)^{3/2}=O(x^2+y^2)$ does not hold either. So $f(x,y)$ is irreducible despite having more than two terms.

We may loosen the conjecture a little, and set about testing whether three, not two, is the maximum number of terms in an irreducible sum of terms. But once again, we are in no luck: the four-term sum $f(x,y) = x^{485}y + x^{477}y^{4} + x^{459}y^{8} + x^{243}y^{32}$, for example, is irreducible. This can be seen by setting, for each term in order, $x = y^{0.7}$, $x = y^{0.31}$, $x = y^{0.21}$, and $x = y^{0.05}$. So conjecturing further seems to have become a little too daunting and perhaps we should back off and consider the possibility that a function may in fact comprise an arbitrarily large number of terms and still be irreducible. Next we prove that this is the case when all the single-variable factors that go in a term are rising power laws, even if the exponents in these laws are all positive integers (i.e., the sum-of-terms function is a polynomial).

\section{An arbitrarily large number of terms} 

Let $f(x,y)$ be such that
\begin{equation*}
f(x,y) = \sum_{i=1}^k x^{a_i} y^{b_i}.
\end{equation*}
If $f(x,y)$ is to be an irreducible function, then clearly no two of the $a_i$'s may equal each other, and similarly no two of the $b_i$'s, since in either case one of the two terms involved would not be independent. Thus, it must be possible to arrange the $a_i$'s into an increasing or decreasing sequence, and similarly the $b_i$'s, but once again for the sake of irreducibility one of the sequences must be increasing while the other is decreasing. We assume $0<a_1<\cdots<a_k$ and $b_1>\cdots>b_k>0$.

For each $i$, we concentrate on valuations of $x$ and $y$ such that $x=y^{z_i}$ for some $z_i>0$, so the $i$th term of $f(x,y)$ is independent if and only if $a_iz_i+b_i>a_jz_i+b_j$ for all $j\neq i$. So arguing for the irreducibility of $f(x,y)$ requires that we find $a_1,\ldots,a_k$, $b_1,\ldots,b_k$, and $z_1,\ldots,z_k$ such that
\begin{equation*}
a_i z_i + b_i > a_j z_i + b_j
\end{equation*}
for all $i$ and all $j\neq i$. Our approach will be to determine the $a_i$'s and the $b_i$'s in such a way as to automatically establish an interval within which to choose the value of each $z_i$. For $1\leq i,j\leq k$, the constraints on this choice are
\begin{eqnarray*}
z_i < (b_i - b_j) / (a_j - a_i) \textrm{ for } i < j,\\
z_i > (b_i - b_j) / (a_j - a_i) \textrm{ for } j < i.
\end{eqnarray*}

A convenient, alternative way to view this condition is to define the ratio
\begin{equation*}
r(i,j) = \frac{b_i - b_j}{a_j - a_i},
\end{equation*}
for which it holds that $r(i,j)=r(j,i)$. Using this equivalence whenever $j<i$ allows the condition to be written as
\begin{equation*}
\max_{1 \leq j < i} r(j,i) < z_i < \min_{i < j \leq k} r(i,j) \textrm{ for } 1 < i < k,
\end{equation*}
in addition to $z_1 < \min_{1 < j \leq k} r(1,j)$ and $z_k > \max_{1 \leq j < k} r(j,k)$. So in order for $f(x,y)$ to be irreducible, we must ensure that
\begin{equation*}
\max_{1 \leq j < i} r(j,i) < \min_{i < j \leq k} r(i,j) \textrm{ for } 1 < i < k.
\end{equation*}

\begin{theorem}
\label{TeoParamArb}
Given any positive integer $k$, $f(x,y)$ is irreducible with $a_i = a_1 (2-\alpha^{i-1})$ and $b_i = b_1 \beta^{i-1}$, where $\alpha$ and $\beta$ are constants such that $0 < \alpha < \beta < 1 - \alpha < 1$.  \end{theorem}

\begin{proof}
We first write $r(i,j)$ as
\begin{equation*}
r(i,j) = \frac{b_1}{a_1} \left(\frac{\beta}{\alpha}\right)^{i-1}\frac{1 - \beta^{j-i}}{1 - \alpha^{j-i}}
\end{equation*}
and note that
\begin{equation*}
r(i,j+1) =  r(i,j) \frac{h_{i,j}(\beta)}{h_{i,j}(\alpha)},
\end{equation*}
where
\begin{equation*}
h_{i,j}(t) = \frac{1-t^{j-i+1}}{1-t^{j-i}},
\end{equation*}
with first derivative given by
\begin{equation*}
h'_{i,j}(t) = \frac{t^{j-i-1} ((j-i)(1-t)-t(1-t^{j-i}))}{(1-t^{j-i})^2}.
\end{equation*}

For $i<j<k$, we have $h'_{i,j}(t)>0$ for $t\in(0,1)$, since using
\begin{equation*}
u_{i,j}(t) = (j-i)(1-t)-t(1-t^{j-i}),
\end{equation*}
we have $u'_{i,j}(t) = (j-i+1)(t^{j-i}-1) < 0$ and $u_{i,j}(1) = 0$, hence $u_{i,j}(t)>0$. It follows that $h_{i,j}(\beta)/h_{i,j}(\alpha) > 1$, and therefore, $r(i,j) < r(i,j+1)$.

Moreover, we also have $r(i-1,k) < r(i, i+1)$ for $1 < i < k$, which can be seen by noting that
\begin{equation*}
r(i-1,k) < \lim_{j \rightarrow \infty} r(i-1,j) < r(i,i+1),
\end{equation*}
where
\begin{equation*}
\lim_{j \rightarrow \infty} r(i-1,j) = \frac{b_1}{a_1} \left(\frac{\beta}{\alpha}\right)^{i-2}
\end{equation*}
and
\begin{equation*}
r(i,i+1) = \frac{b_1}{a_1} \left(\frac{\beta}{\alpha}\right)^{i-1} \frac{1-\beta}{1-\alpha},
\end{equation*}
since $\beta(1-\beta)/\alpha(1-\alpha)>1$ for $\alpha<\beta<1-\alpha$.

The desired inequality follows, since $\max_{1 \leq j < i} r(j,i) =  r(i-1,i) < r(i-1,k) < r(i, i+1) = \min_{i < j \leq k} r(i,j)$.
\end{proof}

For the $a_i$'s and $b_i$'s of Theorem~\ref{TeoParamArb}, following the proof reveals a clear recipe to determine the $z_i$'s: choose $z_1<r(1,2)$, $z_k>r(k-1,k)$, and the remaining ones to satisfy
\begin{equation*}
r(i-1,k)<z_i<r(i,i+1) \textrm{ for } 1<i<k.
\end{equation*}

\section{Integral exponents and constrained valuations}

Our argument in the previous section for the irreducibility of the function $f(x,y)=\sum_{i=1}^k x^{a_i}y^{y_i}$ relied on the particular valuation for $x$ and $y$ that sets $x=y^{z_i}$. All we have required of the exponents $a_i$, $b_i$, and $z_i$ is that they be positive, which leaves plenty of room for them to be nonintegers. This is not a problem in itself, and in fact there exist landmark algorithms that run in time bounded by the input size raised to an irrational power.\footnote{One of the well-known examples of this in the single-variable case is Strassen's algorithm for multiplying two $n\times n$ matrices, whose running time is $O(n^{\log_2 7})$~\cite{Strassen69}.} However, when it comes to the analysis of practical computer algorithms, in most cases we expect the exponents $a_i$ and $b_i$ to be positive integers.

Additionally, depending on the domain at hand, it is often the case that the valuation tying the $x$ and $y$ variables together should only employ values for $z_i$ that are bounded from above by a constant. This is the case of the already noted domain of graph algorithms, in which a graph's numbers $n$ of vertices and $m$ of edges are such that $m = O(n^2)$. In this section we show that there continue to exist exponents for which $f(x,y)$ is irreducible even if we constrain $a_i$ and $b_i$ to be positive integers (hence $f(x,y)$ to be a polynomial), and likewise if $z_i$ is constrained to be no greater than a constant.

\begin{theorem}
\label{TeoParamArbInt}
Given any positive integer $k$, $f(x,y)$ is an irreducible polynomial with $a_i = a_1 (2-\alpha^{i-1})$ and $b_i = b_1 \beta^{i-1}$, where $\alpha=p_\alpha/q_\alpha$ and $\beta=p_\beta/q_\beta$ are rational constants such that $0 < \alpha < \beta < 1 - \alpha < 1$, $a_1 = q_\alpha^{k-1}$, and $b_1 = q_\beta^{k-1}$.  
\end{theorem}

\begin{proof}
Proceed exactly as in the proof of Theorem~\ref{TeoParamArb}, then observe that all $a_i$'s and $b_i$'s are positive integers.
\end{proof}

A detailed example illustrating Theorem~\ref{TeoParamArbInt} is given in Table~\ref{TransformationTable} and Figure~\ref{IntExpGraph} for $k=6$. Table and figure provide different takes on the exact same setting, the former highlighting the integral nature of the exponents in $f(x,y)$ as well as each $a_jz_i+b_j$ as a maximum over all $j$, the latter highlighting each $z_i$ and $r(i,j)$ for $j>i$.

\begin{table}[t]
\caption{For $f(x,y)$ as in Theorem~\ref{TeoParamArbInt}, with $k=6$, $p_\alpha = 1$, $q_\alpha = 3$, $p_\beta = 1$, and $q_\beta=2$, the cell at position $j,i$ gives the value of $a_j z_i + b_j$. These values are highlighted in a bold typeface whenever $j=i$, indicating the maximum on each column (i.e., for fixed $i$ and all $j$).}
\label{TransformationTable}
\centering
\begin{tabular}{c c c c c c c c c}
\hline
& & & \multicolumn{6}{c}{$i$: $z_i$} \\
\cline{4-9}
$j$ & $a_j$ & $b_j$ & 1: 0.05 & 2: 0.14 & 3: 0.21 & 4: 0.31 & 5: 0.47 & 6: 0.70 \\
\hline
1 & 243 & 32 & \textbf{44.15}   & 66.02 & 83.03 & 107.33 & 146.21 & 202.10 \\ 
2 & 405 & 16 & 36.25 & \textbf{72.70} & 101.05 & 141.55 & 206.35 & 299.50 \\ 
3 & 459 & 8 & 30.95 & 72.26 & \textbf{104.39} & 150.29 & 223.73 & 329.30 \\ 
4 & 477 & 4 & 27.85 & 70.78 & 104.17 & \textbf{151.87} & 228.19 & 337.90 \\ 
5 & 483 & 2 & 26.15 & 69.62 & 103.43 & 151.73 & \textbf{229.01} & 340.10 \\ 
6 & 485 & 1 & 25.25 & 68.90 & 102.85 & 151.35 & 228.95 & \textbf{340.50} \\   
\hline
\end{tabular}
\end{table}

\begin{figure}[t]
\centering
\scalebox{0.95}{\includegraphics{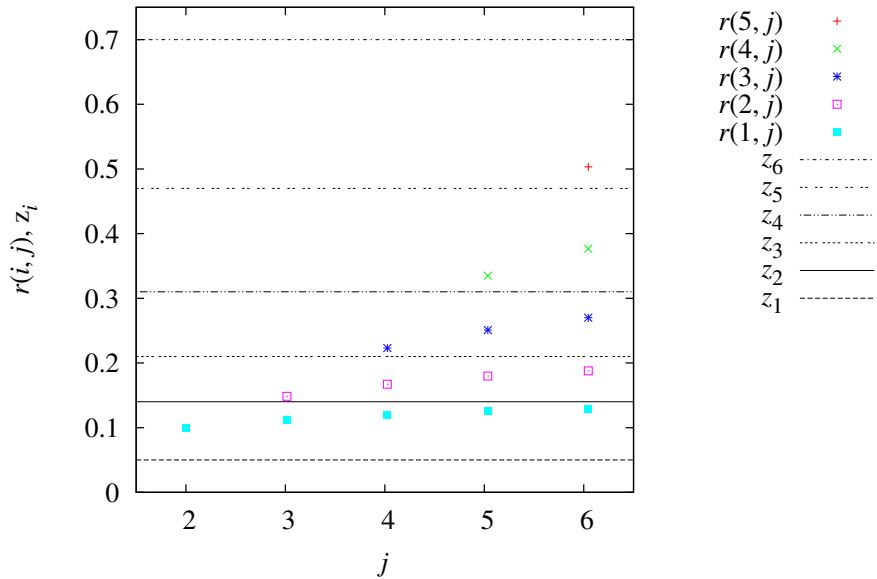}}
\caption{$r(i,j)$ and $z_i$ values for $1 \leq i < j \leq k$, in the same setting as in Table~\ref{TransformationTable}.}
\label{IntExpGraph}
\end{figure}

\begin{theorem}
Given any positive integer $k$ and a constant $c>0$, there exist constants $\alpha$ and $\beta$ such that $0 < \alpha < \beta < 1 - \alpha < 1$ for which $f(x,y)$ is irreducible with $z_i<c$, $a_i = a_1 (2-\alpha^{i-1})$, and $b_i = b_1 \beta^{i-1}$, where $b_1/a_1<c$.    
\end{theorem}

\begin{proof}
Proceed exactly as in the proof of Theorem~\ref{TeoParamArb}, then impose $r(k-1,k)<c$ to obtain an upper bound on the allowable values of $k$:
\begin{equation*}
k < 2 + \log_{\beta/\alpha} \frac{a_1}{b_1} \left(\frac{1 - \alpha}{1 - \beta}\right) c.
\end{equation*}
The result follows from noting that, for $b_1<ca_1$,
\begin{equation*}
\lim_{\beta/\alpha \rightarrow 1^+} \log_{\beta/\alpha} \frac{a_1}{b_1} \left(\frac{1 - \alpha}{1 - \beta}\right) c  = \infty.
\end{equation*}
Therefore, choosing $\alpha$ and $\beta$ to be arbitrarily close to each other accommodates any desired $k$.
\end{proof}

\section{Concluding remarks}

The analysis of computer algorithms via the big-oh notation is an essential part of most activities within computer science, including both theoretical studies and the myriad of applications to which people working in the field devote themselves. In the great majority of situations the algorithm that is being considered is known at some level of detail, so that obtaining big-oh expressions for how much time or space it consumes, though far from being a simple task, is at least a well-defined one. In this article, by contrast, we started out with a ``black-box'' version of an algorithm, that is, a version that we can only analyze by running it on a given set of inputs to make measurements of how much of the necessary resources the algorithm spends.

Faced with the task of discovering big-oh expressions bounding such resource usage, and limiting our search to polynomial-like functions of the relevant variables, we found that, in principle, an automated procedure to carry out the task might have to consider functions comprising an unbounded number of terms. This is surprising, given all the accumulated knowledge on so many algorithms to solve so many different problems, but we feel that it sheds additional light on the big-oh notation itself, especially when we consider the subtle pitfalls that sometimes motivate a deeper examination of its use~\cite{Regan}.

We close with two final remarks. The first is that our conclusions can be easily extended to the case of more variables, recursively by simply fusing together all current variables through appropriate valuations whenever a new variable is added to the pool. The second remark is that, even though for this work we found motivation in the analysis of computer algorithms, the big-oh notation is in fact of much wider interest and applicability, providing a crucial tool whenever it is ``asymptotics,'' not exact figures, that matter. This occurs in several other fields within mathematics, as well as in science and engineering.

\subsection*{Acknowledgments}

The authors acknowledge partial support from CNPq, CAPES, FAPERJ, and a FAPERJ
BBP grant.

\bibliography{bigoh}
\bibliographystyle{unsrt}

\end{document}